\documentclass[letterpaper, 10 pt, conference]{ieeeconf}  

\IEEEoverridecommandlockouts                              

\let\labelindent\relax
\usepackage{fix-cm}
\usepackage{etex}

\usepackage{dblfloatfix}

\usepackage{nag}

\makeatletter
\@ifpackageloaded{xcolor}{}{%
\usepackage[table,x11names,dvipsnames,svgnames]{xcolor}%
}
\makeatother

\usepackage{colortbl}

\usepackage{wrapfig}

\definecolorset{RGB}{lyft}{}{Red,194,39,36;Sunset,202,53,33;Orange,205,68,20;Amber,200,117,42;Yellow,242,169,52;Citron,186,188,44;Lime,112,159,33;Green,56,139,31;Mint,45,118,56;Teal,52,133,135;Cyan,60,132,202;Blue,55,94,248;Indigo,64,13,247;Purple,115,42,248;Pink,176,25,145;Rose,176,32,75}

\usepackage{cite}

\usepackage{microtype}

\usepackage[american]{babel}

\usepackage{array}
\usepackage{multirow}
\usepackage{booktabs}
\usepackage{makecell} 

\newcommand{\myparagraph}[1]{\textbf{\emph{#1}}.}

\ifcsname labelindent\endcsname
\let\labelindent\relax
\fi
\usepackage[inline]{enumitem}

\usepackage{subfig}

\newenvironment{lenumerate}[2][]
{\begin{enumerate}[label=(#2\arabic*),leftmargin=0.2in,itemindent=0.15in,#1]}
{\end{enumerate}}

\setlist*[enumerate,1]{label={\itshape\arabic*)}}

\makeatletter
\newcommand{\paragraphswithstop}{%
\let\copyparagraph\paragraph%
\renewcommand\paragraph[1]{\copyparagraph{##1.}}%
}
\makeatother

\usepackage[framemethod=tikz]{mdframed}

\usepackage{suffix}

\usepackage{environ}

\makeatletter
\newsavebox{\boxifnotempty}
\newcommand{\displayifnotempty}[3]{\sbox\boxifnotempty{#2}\setbox0=\hbox{\usebox{\boxifnotempty}\unskip}%
\ifdim\wd0=0pt
\else
 #1\usebox{\boxifnotempty}#3%
\fi%
}
\newcommand{\ifempty}[2]{\setbox0=\hbox{#1\unskip}%
\ifdim\wd0=0pt%
 #2%
\fi%
}
\newcommand{\ifnotempty}[2]{\setbox0=\hbox{#1\unskip}%
\ifdim\wd0>0pt%
 #2%
\fi%
}
\makeatother

\usepackage{algorithm}

\usepackage{scrlfile}

\makeatletter
\newcommand*\newstoreddef[1]{
  \BeforeClosingMainAux{%
    \immediate\write\@auxout{%
      \string\restoredef{#1}{\csname #1\endcsname}%
    }%
  }%
}
\newcommand*{\restoredef}[2]{
  \expandafter\gdef\csname stored@#1\endcsname{#2}%
}
\newcommand*{\storeddef}[1]{
  \@ifundefined{stored@#1}{0}{\csname stored@#1\endcsname}%
}
\makeatother

\usepackage{pageslts}
\NewEnviron{tee}{\BODY\typeout{Marker Tee [start] ^^J \BODY ^^JMaker Tee [end]}}

\input{preamble/math}
\newcommand{\real}[1]{\mathbb{R}^{#1}{}}

\newcommand{\sphere}[1]{{\mathbb{S}^{#1}}{}}

\newcommand{\defeq}{\doteq}

\DeclarePairedDelimiter{\abs}{\lvert}{\rvert}

\DeclarePairedDelimiter{\norm}{\lVert}{\rVert}

\newcommand{\vct}[1]{\mathbf{#1}}

\DeclareMathOperator*{\argmin}{\arg\!\min}

\DeclareMathOperator{\sign}{sign}

\providecommand{\ve}{\vct{e}}

\providecommand{\vp}{\vct{p}}

\providecommand{\vq}{\vct{q}}

\providecommand{\vs}{\vct{s}}

\providecommand{\vu}{\vct{u}}

\providecommand{\vv}{\vct{v}}

\providecommand{\vx}{\vct{x}}

\providecommand{\vy}{\vct{y}}

\providecommand{\mH}{\vct{H}}
\providecommand{\mI}{\vct{I}}

\providecommand{\mP}{\vct{P}}

\providecommand{\cD}{\mathcal{D}}

\providecommand{\cO}{\mathcal{O}}
\providecommand{\cP}{\mathcal{P}}
\providecommand{\cQ}{\mathcal{Q}}

\providecommand{\cS}{\mathcal{S}}

\providecommand{\cU}{\mathcal{U}}

\providecommand{\trp}{\mathsf{T}}

\usepackage{bm}
\usepackage{mathrsfs}
\usepackage{pgfplots} 
\usepackage{dsfont}

\usetikzlibrary{plotmarks}
\usetikzlibrary{positioning}
\usetikzlibrary{arrows, automata}
\usetikzlibrary{arrows.meta}
\usetikzlibrary{decorations.markings}
\usetikzlibrary{shapes,snakes}
\usetikzlibrary{calc}
\pgfplotsset{plot coordinates/math parser=false} 
\newlength\figureheight 
\newlength\figurewidth 

\newcommand{\phifov}{\phi_{\textrm{FOV}}}

\title{\LARGE \bf Bearing-Only Navigation with Field of View Constraints}

\author{Arman Karimian and Roberto Tron
\thanks{The authors are with the Department of Mechanical Engineering of Boston University, 
Boston, MA. E-mail: {\tt\small \{armandok,tron\}@bu.edu}.}%
}

\begin{document}

\maketitle
\thispagestyle{empty}
\pagestyle{empty}

\begin{abstract} 



  This paper addresses the problem of navigation using only relative direction measurements (i.e., relative distances are unknown) under field of view constraints. We present a novel navigation vector field for the bearing-based visual homing problem with respect to static visual landmarks in 2-D and 3-D environments. Our method employs two control fields that are tangent and normal to ellipsoids having landmarks as their foci. The tangent field steers the robot to a set of points where the average of observed bearings is parallel to the average of the desired bearings, and the normal field uses the angle between a pair of bearings as a proxy to adjust the robot's distance from landmarks and to satisfy the field of view constraints. Both fields are blended together to construct an almost globally stable control law. Our method is easy to implement, as it requires only comparisons between average bearings, and between angles of pairs of vectors. We provide simulations that demonstrate the performance of our approach for a double integrator system and unicycles.

\end{abstract}

\section{INTRODUCTION}
Control in robotics systems is an ongoing research area in many respects. One interesting problem is the one of bearing-only navigation, motivated by the use of vision sensors in robotics applications. Such sensors, such as monocular cameras, can provide accurate bearing (relative direction) measurements, although the corresponding distances are typically difficult to obtain with comparable precision. In addition, vision sensors typically have a limited field of view (FOV). These two limitations increases the complexity of navigation considerably.

The problem addressed in this paper is visual homing, the task of reaching a desired location using the bearing measurements of fixed landmarks in the surrounding environment \cite{hong1992image}. A practical application of this problem is when a robot takes a picture of the environment from a home location, moves to a new location, and then needs to return to the home location using bearing vectors extracted from the current and home images. 
Existing methods can be divided into gradient methods~\cite{tron2014optimization, lambrinos1998landmark, cowan2002visual}, image-based visual servoing~\cite{corke2003mobile, papanikolopoulos1993adaptive, liu2013visual}, and ad-hoc methods~\cite{argyros2001robot, lim2009robust, liu2010bearing, loizou2007biologically}. While gradient methods can achieve global stability \cite{tron2014optimization} without using range estimation (actual or estimated), FOV constraints are not generally considered in these approaches, and it is common to assume omnidirectional vision sensors. Notable exceptions of~\cite{lopez2009homography}, where a homography-based approach is given for keeping a single target in the field of view of a unicycle with an onboard IMU and with a camera attached to the body, and~\cite{cowan2002visual} where a navigation function based approach was used but required planar targets with known geometry.

FOV constraints restrict the feasible locations for the moving sensor carried by the robot, thus effectively creating obstacles in the configuration space of the robot. The goal then becomes to complete the visual homing task while avoiding these obstacles; however, the location of such obstacles is not directly available to the robot, due to the missing distance measurements. Many different methods have been suggested for control-based obstacle avoidance, some of which are potential methods~\cite{khatib1986real, hernandez2011convergence}, navigation functions~\cite{rimon1992exact}, and harmonic functions~\cite{connolly1990path,feder1997real, kim1992real}. Potential methods are prone to local minima away from the goal point; navigation functions are free of local minima but are sensitive to the value of a tuning parameter which is not known a priori, and harmonic functions are usually computationally demanding and require the location of the obstacles. An alternative approach for obstacle avoidance is to directly design a \emph{navigation vector field} which encodes the objectives (desired home location and FOV obstacle avoidance), and is employed directly or indirectly in the control synthesis step. This idea has been previously used for obstacle avoidance in unicycles, but with full information on the relative position between robot and the obstacles~\cite{panagou2014motion, panagou2016distributed}.


\noindent\myparagraph{Our approach}
We introduce two orthogonal flows that respectively adjust the direction of the average of the bearings, and the angle between a pair of bearings. We then combine these two flows into a navigation flow, which in turn is used to design controllers for solving the visual homing problem in the presence of FOV constraints laws with damped double integrators and unicycles.
Our approach is applicable to both 2-D and 3-D environments, and presents almost-global convergence (the integral curve of every starting point, except for a set of Lebesgue measure zero, converges to the home location). Our approach does not rely on all the bearings directly, but it rather uses the normalized average of the bearings, and a single angle between two non-collinear bearings.
We assume that the camera on the robot can rotate independently from the body and direction of motion of the robot, and we model the field of view as a cone with angle less than $\pi$. Additionally, we assume that robot's local reference frame is axis-aligned with a fixed world frame (e.g., through a global compass direction).

\section{NOTATION AND PRELIMINARIES}
We denote the dimension of the workspace by $d\in\{2,3\}$. 
The \emph{bearing measurement} vector between two distinct points $\vx_i, \vx_j\in\real{d}$ is given by the unit vector:
\begin{equation}
\vct{u}(\vx_i,\vx_j) \defeq \dfrac{\vx_j-\vx_i}{\norm{\vx_j-\vx_i}}\,.
\label{eq:bearingdef}
\end{equation}
We denote the cardinality of a discrete set $\cP$ as $\abs{\cP}$, and the boundary of a continuous set $\cQ$ as $\partial\cQ$. 
The identity matrix is denoted by $\mI_d\in\real{d\times d}$, the $d$-dimensional unit sphere by $\mathbb S^d$, and the Minkowski sum by $\oplus$. We use $\measuredangle(\vu_1,\vu_2)$ to denote the (non-oriented) angle in radians between two vectors $\vu_1,\vu_2$.
A \emph{projection matrix} $\mP(\vct v)\in\real{d\times d}$ for a vector $\vct v\in\real{d}$ is defined by:
\begin{equation}\label{eq:projection}
\mP(\vct v)\defeq \mI_d - \frac{\vct v \vct v^\trp}{\|\vct v\|^2};
\end{equation}
$\mP(\vct v)$ is symmetric, positive semidefinite, with a zero eigenvalue corresponding to $\vct v$, while other eigenvalues are one.
\begin{lemma}\label{lemma:unitvec}
Given a unit vector in the from $\vy=\frac{\bm g(\vx)}{\norm{\bm g(\vx)}}$, the Jacobian of $\vy$ with respect to $\vx$ is given by: $\norm{\bm g(\vx)}^{-1}\mP(\bm g(\vx))\frac{\partial\bm g}{\partial \vx}$.
\end{lemma}
\begin{proof}
Since $\frac{\partial\vy}{\partial\vx}=\norm{\bm g(\vx)}^{-1}\frac{\partial\bm g}{\partial\vx}+\norm{\bm g(\vx)}^{-3}\bm g\bm g^\trp\frac{\partial\bm g}{\partial\vx}$, the proof is complete by collecting $\norm{\bm g(\vx)}^{-1}$ and $\frac{\partial\bm g}{\partial\vx}$.
\end{proof}

Given $k$ fixed and distinct points $\cP=\{\vp_i\}_{i=1}^k$ in $\real{d}$, we define the \emph{distance function} $\vartheta(\vx)\defeq\sum_{i=1}^k\norm{\vx-\vp_i}$ to be the sum of distances from $\vx$ to all points in $\cP$. We have the following facts regarding the function $\vartheta$.
\begin{definition}\label{def:kellipsoid}
A \emph{$k$-ellipsoid} is the set of points over which $\vartheta(\cdot)$ is equal to a constant $r$, and points in $\cP$ are called \emph{foci}. it can be also defined as the boundary of the set-valued map: 
\begin{equation}
\cQ(r)=\{\vx\in\real{d}:\vartheta(\vx)\leq r\}\,.
\label{eq:kellipsoid}
\end{equation}
\end{definition}
\begin{lemma}\label{lemma:hessian}
 Hessian of $\vartheta$ is positive semidefinite, and is positive definite if all points are not collinear.
\end{lemma}
\begin{proof}
The gradient of $\vartheta(\cdot)$ is given by $\frac{\partial\vartheta}{\partial\vx}=-\sum_{i=1}^{k}\vu(\vx,\vp_i)$ and its Hessian is given by $\mH(\vx)\defeq\frac{\partial^2\vartheta}{\partial\vx^2}=\sum_{i=1}^{k}\frac{\mP(\vu(\vx,\vp_i))}{\norm{\vx-\vp_i}}$. Since $\mH$ is the sum of positive semidefinite matrices, $\mH$ is also positive semidefinite. Moreover, we have $\vv^\trp\mH\vv=\vct 0$, and hence $\mH$ is positive semi-definite, if and only if terms in the sum has the same eigenvector with zero eigenvalue, i.e., all the points in $\cP$ are collinear, and $\vx$ lies on the same line. In all other cases, $\mH$ is positive definite.
\end{proof}

A point $\vp\in\real{d}$ is said to be a geometric median of the set $\cP$ if $\vp\in\argmin_{\vx}\vartheta(\vx)$. We have $\frac{\partial\vartheta}{\partial\vx}(\vp)=\vct 0$ if $\vp\notin\cP$.
\begin{lemma}[\!\!\cite{vardi2000multivariate}]
The geometric median of set $\cP$ is unique, unless all points in $\cP$ are collinear and $k$ is even.
\label{lemma:v}
\end{lemma}

\section{BEARING-ONLY NAVIGATION}
\begin{figure*}
\centering
\hfill
\begin{minipage}[c]{0.4\textwidth}
\subfloat[Obstacle sets $\cO_{12}$ for $\phifov\in\{{\color{TealBlue}\frac{3\pi}{8}},{\color{NavyBlue}\frac{7\pi}{8}}\}$ and the desired set {\color{ForestGreen}$\cD_{12}^*$}. These sets are obtained by joining two circular segments.\label{fig:demon}]{\includegraphics{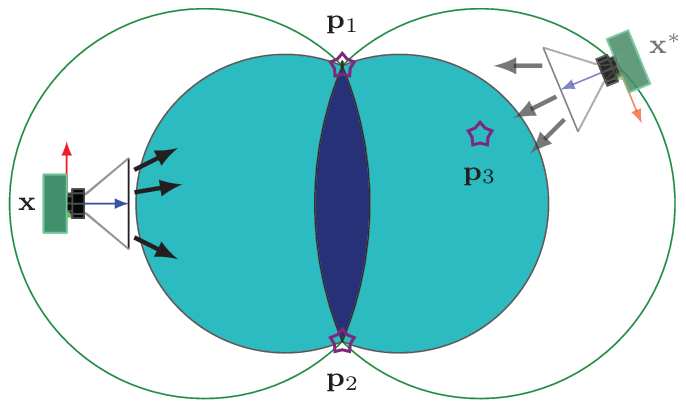}}
\end{minipage}
\hfill
\begin{minipage}[c]{0.5\textwidth}
\subfloat[$\bm f_t^\circ(\vx), k=2$]{\includegraphics{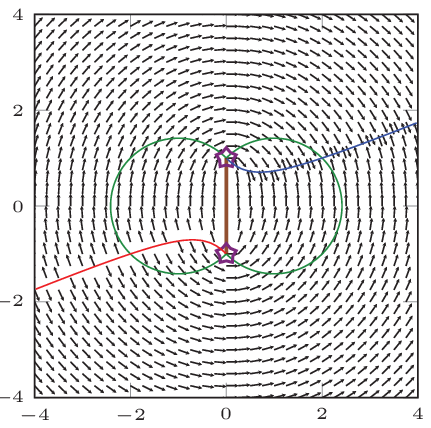}\label{fig:elliptical_flow2}} 
\subfloat[$\bm f_n^\circ(\vx), k=2$]{\!\!\!\!\!\!\!\!\includegraphics{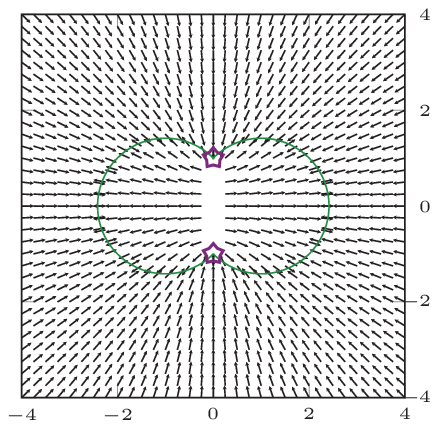}\label{fig:normal_flow2}} \\
\subfloat[$\bm f_t^\circ(\vx), k=3$]{\includegraphics{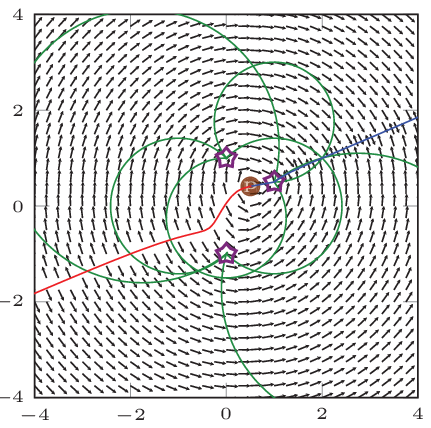}\label{fig:elliptical_flow3}}
\subfloat[$\bm f_n^\circ(\vx), k=3$]{\!\!\!\!\!\!\!\!\includegraphics{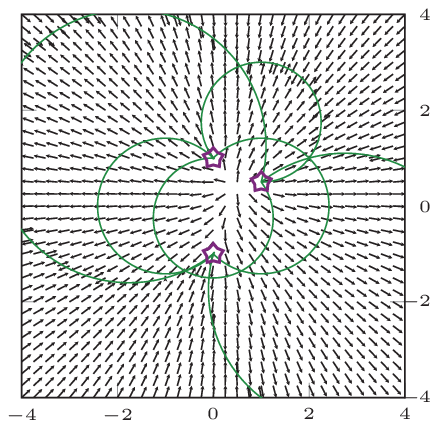}\label{fig:normal_flow3}}
\end{minipage}
\hfill
\caption{\protect\subref{fig:demon} shows the current and goal positions of the robot with respect to three landmarks\! (\protect\tikz\protect\node[star,star points=5,draw, thick, violet, fill=none, scale=0.55](p) at (0,0) {};), as well as $\cO_{12}$ and $\cD_{12}^*$ sets. In \protect\subref{fig:elliptical_flow2}, the ellipsoidal flow from $\bm f_t$ is observed, and the isonormal curves {\color{blue}$\bm\xi_{\vv^*}$} and {\color{red} $\bm\xi_{-\vv^*}$} for the first two landmarks. The goal position $\vx^*$ is at the intersection of the {\color{ForestGreen}$\cD_{12}^*$} set and the {\color{blue}$\bm\xi_{\vv^*}$} curve. Since there are an even number of collinear landmarks, geometric median is not unique and is shown by the line segment connecting $\vp_1$ to $\vp_2$. In \protect\subref{fig:normal_flow2}, the normal flow from $\bm f_n$ is shown. In \protect\subref{fig:elliptical_flow3}-\protect\subref{fig:normal_flow3}, the two flows are given for three landmarks, with $\vx^*$ at the intersection of {\color{ForestGreen}$\cD_{12}^*$}, {\color{ForestGreen}$\cD_{23}^*$}, {\color{ForestGreen}$\cD_{13}^*$} and {\color{blue}$\bm\xi_{\vv^*}$}. In this case, since the three landmarks are not collinear, geometric median is unique and is shown by \protect\tikz\protect\node[circle, draw=Sepia, fill=Sepia, inner sep=0pt, text=white, minimum size=2pt, scale = 1.7] at (0,0) { \tiny$\vct p$};.}
\label{fig:demonstration}
\end{figure*}
We address the visual homing problem with restricted FOV:
\begin{problem}
Given a set of $k$ landmarks $\cP=\{\vp_i\}_{i=1}^k$ in $\real{d}$ with $k\geq 2$, find a controller that steers the robot to the desired position $\vx^*\in\real{d}$ specified by a set of desired relative bearings $\{\vu^*_i\}_{i=1}^k$ (where $\vu^*_i\defeq\vu(\vx^*,\vp_i)$), using only relative bearing measurements of the landmarks with respect to robot's current location $\vu_i\defeq \vu(\vx,\vp_i)$, while satisfying the following pairwise conditions at all times: 
\begin{equation}
\measuredangle(\vu_{i},\vu_{j})< \phifov, \quad \forall i,j\in\{1,\dots,k\}\,.
\label{eq:FOV}
\end{equation}
\label{problem1}
\end{problem}

The constraints given in \eqref{eq:FOV} ensure that all the landmarks remain in the visual field of the robot's camera observing them, modeled as a cone with angle $\phifov<\pi$. This can be summarized as $\vx(t)\notin\bigcup_{i,j}\cO_{ij}$ for all $t$, where $\cO_{ij}$ is a \emph{field of view obstacle} set defined as:
\begin{equation}
\cO_{ij} = \{\vx\in\real{d} : \measuredangle(\vct u_{i},\vct u_{j})\geq \phifov\}.
\end{equation}
On the boundary of $\cO_{ij}$, the robot's view angle of the landmarks $\vp_i$ and $\vp_j$ is equal to $\phifov$. While these regions are to be avoided, there are desired regions with shape similar to the boundary of $\cO_{ij}$ where the view angle of landmarks is equal to the desired view angle given by desired bearings. We define such \emph{desired field of view} as:
\begin{equation}
\cD^*_{ij} = \{\vx\in\real{d} : \measuredangle(\vct u_{i},\vct u_{j})= \measuredangle(\vu^*_{i},\vu^*_{j})\}.
\end{equation}
It follows that $\vx^*\in\bigcap_{i,j}\cD^*_{ij}$. These sets are depicted in Fig.~\ref{fig:demon} for $d=2$. For $d=3$, $\cO_{ij}$ and $\cD^*_{ij}$ can be visualized by revolving their 2-D version about the line intersecting with $\vp_i$ and $\vp_j$.



\subsection{Tangential and normal vector fields}
Here we define two orthogonal vector fields that are then combined into a single navigation field. Let $\vv,\vv^*\in\real{d}$ be the normalized sum of the current and desired bearings:
\begin{equation}
\begin{aligned}
\vv(\vx)&=
\frac{\sum_{i=1}^{k}\vu_i}{\norm{\sum_{i=1}^{k}\vu_i}},\\ 
\vv^*&=\frac{\sum_{i=1}^{k}\vu_i^*}{\norm{\sum_{i=1}^{k}\vu^*_i}}=\vv(\vx^*).
\end{aligned}
\end{equation}
Notice that $\vv$ is a vector field and $\vv^*$ is a fixed value.
\begin{remark}
Vector $\vv$ is defined everywhere, except at the landmarks (i.e. $\vx=\vp_i$) since the corresponding bearing vector $\vu_i$ is undefined, and where the gradient of $\vartheta$ is zero ($-\sum_{i=1}^{k}\vu_i=\vct 0$), which happens at the geometric median of the landmarks. As stated in Lemma \ref{lemma:v}, the geometric median is unique unless we have an even number of collinear foci (e.g. an ellipse), in which case $\vv$ is not defined on the line segment that contains the two middle foci \cite{kemperman1987median}.
\end{remark}

Let $\delta_{ij}$ be the difference between the cosine of the current and the desired bearings of the landmarks $\vp_i$ and $\vp_j$:
\begin{equation}
\delta_{ij} \defeq \vu_i^\trp\vu_j-\vu_i^{*\trp}\vu_j^* \,.
\label{eq:delta}
\end{equation}
Notice that $\measuredangle(\vu_{i},\vu_{j})-\measuredangle(\vu^*_{i},\vu^*_{j})$ and $\delta_{ij}$ have opposite signs, and $\delta_{ij}$ is a function of $\vx$. We define a \emph{tangential field} $\bm f_t(\vx)$ and a \emph{normal field} $\bm f_n(\vx)$ as:
\begin{subequations}
\label{eq:flow}
\begin{align}
\bm f_t(\vx) &= -\mP(\vv)\vv^*, \label{eq:flowt}\\
\bm f_n(\vx) &= \sign(\delta_{\hat{\imath}\hat{\jmath}})\vv, \label{eq:flown}
\end{align}
\end{subequations}
where indices $\hat\imath$ and $\hat\jmath$ are chosen as:
\begin{equation}
  \hat{\imath},\hat{\jmath} = \argmin_{\substack{i,j \\ \vu^*_i\neq\vu^*_j}} \delta_{ij}.
  \label{eq:chosenpair}
\end{equation}

These two vector fields are the building blocks of our navigational vector field, which we present in Section~\ref{sec:combined}. In the remainder of this subsection, we instead focus on the two vector fields individually, and discuss their behavior. First, we show that these two vector fields are always orthogonal to each other. Then, we show that $\bm f_t$ maintains the sum of distances from landmarks and forces $\vv$ to reach $\vv^*$, while $\bm f_n$ uses the angle between the bearings as a proxy for the distance from the home location, and adjusts the actual distance from the landmarks.

\begin{lemma}\label{lemma:ortho}
The two vector fields $\bm f_t$ and $\bm f_n$ are orthogonal.
\end{lemma}
\begin{proof}
Since $\vv^\trp\mP(\vv)=\vv^\trp(\mI_d-\vv\vv^\trp)$ is zero, we have $\bm f_n^\trp\bm f_t= 0$ everywhere.
\end{proof}

\begin{lemma}
The integral curves of $\bm f_t$ correspond to $k$-ellipsoids.
\end{lemma}
\begin{proof}
Following Definition~\ref{def:kellipsoid}, a $k$-ellipsoid is the set of points over which the distance function $\vartheta(\vx)$ is constant. The gradient of $\vartheta$, $-\sum_{i=1}^k\vu_i$, is parallel to $\vv$. Since $\bm f_t$ is always orthogonal to $\vv$, then $\vartheta$ is constant in the direction of $\bm f_t$, hence the claim.
\end{proof}

In the following, we use a navigation field to determine a moving direction, without considering its magnitude; with this motivation, we introduce the following:

\begin{definition}
We denote the normalized version of a vector field $\bm f(\cdot)$ by $\bm f^\circ\defeq \frac{\bm f}{\norm{\bm f}}$, with $\bm f^\circ=\vct 0$ whenever $\norm{\bm f}=0$ or $\bm f$ is undefined. 
\end{definition}
Examples of the normalized versions of the vector fields~\eqref{eq:flow} for $k=2$ and $k=3$ landmarks and $d=2$ are shown in Fig.~\ref{fig:demonstration}. The tangential flows in Fig.\ref{fig:elliptical_flow2} and Fig.\ref{fig:elliptical_flow3} lie on 2-ellipsoids (or simply ellipses) and 3-ellipsoids respectively. The normal flows in Fig.\ref{fig:normal_flow2} and Fig.\ref{fig:normal_flow3}, intuitively, move away from the landmarks if a negative $\delta_{\hat{\imath}\hat{\jmath}}$ exists (i.e., when $\measuredangle(\vu_{\hat{\imath}},\vu_{\hat{\imath}})>\measuredangle(\vu_{\hat{\imath}}^*,\vu_{\hat{\imath}}^*)$), and move closer to the landmarks if all $\delta_{ij}$ values are positive until the angle between one pair of bearings is equal to the angle between their corresponding desired bearings (i.e. $\delta_{\hat{\imath}\hat{\jmath}}=0$ and for other pairs, $\delta_{ij}\geq 0$). Let $\cD^*$ be the boundary of the union of all $\cD_{ij}^*$ sets and their interior points; then $\delta_{\hat{\imath}\hat{\jmath}}$ is zero on $\cD^*$, negative inside of it, and positive outside of it (see Fig.~\ref{fig:normal_flow3}). The condition $\vu_i^*\neq\vu_j^*$ in \eqref{eq:chosenpair} ensures that the angle between $\vu_i$ and $\vu_j$ can be used as a measure for adjusting the distance of camera to landmarks, since if equality holds it means that the corresponding landmarks $\vp_i$ and $\vp_j$ and the goal position $\vx^*$ all lie on a line and moving closer or away from the landmarks on that line does not change the value of $\delta_{ij}$. This assumption is not overly restrictive, since coincident bearings represent redundant constraints for defining the home location, and are not practically useful. We show here that moving closer towards landmarks by $\bm f_n$ increases the overall pairwise angles, and vice versa. 

\begin{theorem}\label{theorem:away}
Moving in the direction $\vv$ decreases the sum of the cosine of pairwise angles $\measuredangle(\vu_i,\vu_j)$.
\end{theorem}
\begin{proof}
Take the Lyapunov function $V=\frac{1}{2}\norm{\sum_{i=1}^{k}\vu_i}^2$. By expanding $V$, we rewrite it as $V=\frac{k}{2}+\sum_{i<j}\vu_i^\trp\vu_j$, which is the sum of cosine of all pairwise angles between the landmarks plus a constant term. Taking the gradient, $\frac{\partial V}{\partial \vx}=-\sum_{i<j}\big(\frac{\mP(\vu_i)}{\norm{\vx-\vp_i}}\vu_j+\frac{\mP(\vu_j)}{\norm{\vx-\vp_j}}\vu_i\big)$ and since $\mP(\vu_i)\vu_i=\vct 0$, we have $\frac{\partial V}{\partial \vx}=-\sum_{i<j}(\frac{\mP(\vu_i)}{\norm{\vx-\vp_i}}+\frac{\mP(\vu_j)}{\norm{\vx-\vp_j}})(\vu_i+\vu_j)=-\frac{1}{2}(\sum_{i=1}^k \frac{\mP(\vu_i)}{\norm{\vx-\vp_i}})(\sum_{i=1}^k\vu_i)$. Taking the flow $\dot{\vx}=\vv$, we get $\dot{V}=\frac{\partial V}{\partial \vx}\dot{\vx}=-\vv^\trp\mH\big(\sum_{i=1}^k\vu_i\big)$ which is always negative ($\mH$ is the Hessian from Lemma~\ref{lemma:hessian}) unless at geometric median ($\sum_{i=1}^k\vu_i=\vct 0$), or when all foci are collinear along the line containing foci.
\end{proof}
The behavior induced by $\bm f_n$ forces the robot to converge to $\cD^*$. Since the obstacle sets $\cO_{ij}$ are contained in the interior of $\cD^*$, staying on $\cD^*$ or its exterior guarantees collision avoidance with the FOV obstacle sets. Even though we showed in Theorem~\ref{theorem:away} that moving away from the landmarks reduces the overall angles, if the robot has a small field of view (i.e. $\phifov$ is small) and starts from the interior of $\cD^*$ (i.e. $\delta_{\hat\imath\hat\jmath}<0$) close to landmarks, it might enter an obstacle set on its way towards $\cD^*$. Even if such case happens, one remedy for it would to keep moving away in the same direction until the landmarks are back in sight. Due to boundedness of the $\cO_{ij}$ sets, there should exist a $k$-ellipsoid $\partial\cQ(r)$ for some positive $r$ that encapsulates all the $\cO_{ij}$ sets.

The behavior of the tangential flow $\bm f_t$ is less intuitive. First, we show that $\bm f_t$ yields convergence of $\vv$ to $\vv^*$. Then, we investigate its equilibrium points.
\begin{theorem}\label{thm:ft}
Following the flow $\dot{\vx} = \bm f_t(\vx)$ leads to convergence to $\vv=\vv^*$ almost globally.
\end{theorem}
\begin{proof}
Take the $V=\frac{1}{2}\norm{\vv-\vv^*}^2$ as Lyapunov function. Using the chain rule, $\frac{\partial V}{\partial \vx} =(\vv-\vv^*)^\trp\frac{\partial \vv}{\partial \vx}$ and $\frac{\partial\vv}{\partial\vx}=\norm{\sum_{i=1}^k \vu_i}^{-1}\mP(\sum_{i=1}^k \vu_i)\frac{\partial \sum_{i=1}^k\vu_i}{\partial\vx}$ (Lemma~\ref{lemma:unitvec}). From Lemma~\ref{lemma:hessian}, we have $\frac{\partial \sum_{i=1}^k\vu_i}{\partial\vx}=-\mH$, where $\mH$ is the Hessian of $\vartheta$. Also from the definition of \eqref{eq:projection} we have $\mP(\sum_{i=1}^k\vu_i)=\mP(\vv)$. Hence, $\dot{V}=\frac{\partial V}{\partial \vx}\dot{\vx}=\norm{\sum_{i=1}^k \vu_i}^{-1}(\vv-\vv^*)^\trp\mP(\vv)\mH\mP(\vv)\vv^*$. Since $\vv^\trp\mP(\vv)=\vct 0$ and $\mP(\cdot)$ is symmetric, we can simplify the expression for $\dot{V}$ as $\dot{V}=-\norm{\sum_{i=1}^k \vu_i}^{-1}\bm f_t^\trp\mH\bm f_t$. We know from Lemma~\ref{lemma:hessian} that $\mH$ is positive definite, or positive semidefinite when all landmarks are collinear (in which case all bearing vectors are parallel). For the first case, we have $\dot{V}<0$ whenever $\bm f_t\neq\vct 0$. For the second case, $\vv$ is also parallel to all bearings and the zero eigenvector of $\mH$, and because $\bm f_t$ is orthogonal to $\vv$ (see Lemma~\ref{lemma:ortho}) we again have $\dot{V}<0$. However, we have $\bm f_t=\vct 0$ if $\vv=\vv^*$ or $\vv=-\vv^*$. It means that any point $\vx_0$ with $\vv(\vx_0)=-\vv^*$ is an unstable equilibrium point of $\bm f_t$.
\end{proof}
In order to find the equilibrium points of $\bm f_t$, we need to find points where $\vv=\vv^*$ or $\vv=-\vv^*$, since $\mP(\vv^*)\vv^*=\mP(-\vv^*)\vv^*=\vct 0$. In a more general approach, it would be beneficial to know the properties of the set of points where $\vv=\vv_0$ for a given unit vector $\vv_0$. We will show that such points form a curve, which we call an \emph{isonormal curve}, and start from a/the geometric median point of the foci and moves away from the foci towards infinity.
\begin{proposition}
Let $\bm\xi_{\vv_0}=\{\vx \in \real{d}:\vv(\vx)=\vv_0\}$ be the set of points where $\vv$ is equal to a value $\vv_0\in\sphere{d-1}$. Then:
\begin{enumerate}
\item $\bm\xi_{\vv_0}$ is a 1-D open curve.
\item Each focus $\vp_i$ belongs to $\bm\xi_{\vv_0}$ for any $\vv_0\in\cU_i$, where
\begin{equation}
\cU_i\defeq\{\frac{\ve}{\norm{\ve}}: \ve\in\sum_{j\neq i}\vu(\vp_i,\vp_j)\oplus\mathbb{S}^{d-1}\}.
\end{equation}
\item Every point in $\bm\xi_{\vv_0}\setminus \cP$ is regular.
\end{enumerate}
\end{proposition}
\begin{proof}
  A $k$-ellipsoid $\cS=\partial\cQ(r)$ is the boundary of the convex and compact set $\cQ(r)$. Moreover, due to positive definiteness of the hessian $\mH$, we have that $\cQ$ is strictly convex (take any two points $\vs_1$, $\vs_2$ on $\cS$, due to strict convexity the value of $\vartheta$ is strictly less than $r$ on the line segment between $\vs_1$ and $\vs_2$). This is also true when $\mH$ is positive semidefinite (i.e. when all landmarks are contained on a line $\ell$) since in this case $\mH$ is positive definite everywhere except along $\ell$, and since the value of $\vartheta$ changes along $\ell$, $\cS$ does not contain any line segments. Therefore, $\cQ$ is strictly convex.\\
Due to compactness and strict convexity of $\cQ$, every linear function has one point of minimum and one point of maximum, which means that each vector $\vv_0$ is normal to $\cS$ at exactly two points. By restricting normal vectors to be towards the inside of $\cS$ (the opposite direction of gradient of $\vartheta$), for a given normal vector $\vv_0$, there exists only a single point on $\cS$ where $\vv_0$ is normal to $\cS$; equivalently stated, the Gauss map \cite{gauss2005general} of $\cS$ (i.e. $\cS\mapsto\mathbb S^{d-1}$) is surjective.\\
Uniqueness of the normal vector $\vv_0$ on $\partial\cQ$ implies that $\bm\xi_{\vv_0}$ intersects with any $k$-ellipsoid ($\partial\cQ(r)$) at a single point. Starting from $r_{\textrm{min}}=\min_\vx\vartheta(\vx)$, which happens at a/the geometric median point, as $r\rightarrow\infty$ we have an infinite number of points in $\bm\xi_{\vv_0}$.
Here we show that these points form a 1-D curve which is regular everywhere except at the foci. Let the function $\bm\zeta(r)\in\bm\xi_{\vv_0}$ represent the point on $\cS$ with $\vv(\bm\zeta)=\vv_0$ and $\vartheta(\bm\zeta)=r$. By definition, we have $\bm\eta\defeq\sum_{i=1}^k\vu(\bm\zeta,\vp_i)\propto \vv_0$. If we change $\bm\zeta$ such that $\frac{\partial}{\partial\bm\zeta}\bm\eta\propto\vv_0$ then $\bm\zeta+\partial\bm\zeta$ remains in $\bm\xi_{\vv_0}$. By taking the derivative, we see that $-\mH(\bm\zeta)\partial\bm\zeta$ needs to be proportional to $\vv_0$, or equivalently $\partial\bm\zeta\propto-\mH(\bm\zeta)^{-1}\vv_0$. Regularizing with respect to $r$, we get $\frac{\partial\bm\zeta}{\partial r}=-\mH(\bm\zeta)^{-1}\vv_0(\bm\eta^\trp\mH(\bm\zeta)^{-1}\vv_0)^{-1}$. This derivative exists everywhere except at foci. Therefore $\bm\xi_{\vv_0}$ is a curve that is regular everywhere except at foci.\\
Now, we discuss when a focal point $\vp_i$ is contained in any curve $\bm\xi_{\vv_0}$. It is straightforward to see that the set $\cU_i$ contains the limits of $\vv(\vx)$ as $\vx\rightarrow\vp_i$. 
\end{proof}


Intuitively, each isonormal curve $\bm\xi_{\vv_0}$ starts from geometric median and move away from the foci such that they only intersects once with the boundary of each $\cQ(r)$ set. Fig.~\ref{fig:ellipsoid} depicts these curves for different unit vectors $\vv_0$ with four foci. See~\cite[Fig. 2]{karimian2020bearing} for the case with three landmarks.
\begin{remark}
Since the unstable equilibrium points of $\bm f_t$ lie on the $\bm\xi_{-\vv^*}$ curve, $\bm f_t$ is almost globally stable, except on the 1-D set $\bm\xi_{-\vv^*}$ which has a Lebesgue measure zero.
\end{remark}
\begin{figure}[t]\centering
\scalebox{1}{\includegraphics{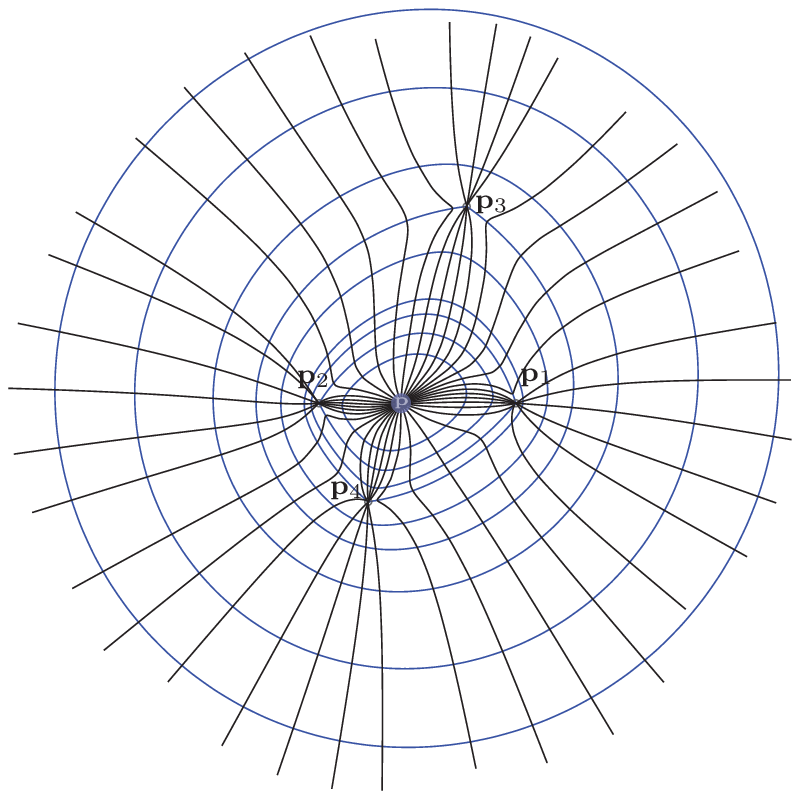}}
\caption{Multiple confocal $4$-ellipsoids in 2-D with various isonormal curves. Point \protect\tikz\protect\node[circle, draw=blue, fill=blue, inner sep=0pt, text=white, minimum size=2pt, scale = 1.7] at (0,0) { \tiny$\vct p$}; is the geometric median of foci.}
\label{fig:ellipsoid}
\end{figure}
\subsection{Combined flow} \label{sec:combined}
In this subsection, we combine the normalized flows $\bm f_t^\circ$ and $\bm f_n^\circ$ together into our final navigational vector field. 
In order to introduce smooth transitions, we use the smooth bump function of degree three defined as:
\begin{equation}
b_\epsilon(x)\defeq
\begin{cases}
0 & x\leq 0 \\
\sum_{i=0}^{3}a_ix^i & 0\leq x\leq \epsilon \\
1 &  \epsilon\leq x
\end{cases}
\label{eq:bump}
\end{equation}
where $a_0=a_1=0$, $a_2=3\epsilon^{-2}$, $a_3=-2\epsilon^{-3}$, and $\epsilon$ is a design parameter.

The idea behind combining the two orthogonal flows $\bm f_t^\circ,\bm f_n^\circ$ is to let $\bm f_t^\circ$ steer the value of $\vv$ close to $\vv^*$ (i.e. take the robot close to the $\bm\xi_{\vv^*}$ curve), and to adjust the view angle by using $\bm f_n^\circ$ only when $\vv$ is close enough to $\vv^*$ if $\delta_{\hat\imath\hat\jmath}>0$. This is because staying on $\cD^*$ might take the robot too close to the landmarks, or cause collision with them (see Fig.~\ref{fig:normal_flow2} and Fig.~\ref{fig:normal_flow3}). However, if $\delta_{\hat\imath\hat\jmath}$ is negative, we always want the robot to move away from the landmarks to avoid entering the obstacle set $\cO_{\hat\imath\hat\jmath}$.  
Here, we introduce an example for how to combine these two flows in order to achieve this desired behavior. The combined field is given by:
\begin{equation}
\bm f(\vx) \defeq g_t(\vx)\bm f_t^\circ(\vx) + g_n(\vx)\bm f_n^\circ(\vx)\,,
\end{equation}
where $g_t$ and $g_n$ are \emph{gain functions} with non-negative values:
\begin{subequations}
\label{eq:flowadj}
\begin{align}
g_t(\vx)&\defeq \min(1,\sqrt{1-\vv^\trp\vv^*})\label{eq:flowadjt}\\
g_n(\vx)&\defeq \max(0,\vv^\trp\vv^*)b_\epsilon(\delta_{\hat\imath\hat\jmath})+b_\epsilon(-\delta_{\hat\imath\hat\jmath})\,.\label{eq:flowadjn}
\end{align}
\end{subequations}
The tangential gain function $g_t$ is:
\begin{enumerate*}
\item equal to one when $\measuredangle(\vv,\vv^*)\geq \frac{\pi}{2}$,
\item less than one as $\vv$ gets close to $\vv^*$, i.e., when $\measuredangle(\vv,\vv^*)< \frac{\pi}{2}$.
\end{enumerate*}
This behavior prioritizes convergence to $\vv^*$ when the goal position and starting positions are on different sides of the landmarks, and also slows down the convergence rate to $\vv^*$ as $\vv\rightarrow\vv^*$.
\begin{figure*}[]
\centering
\hfill
\subfloat[Double integrator, $\lambda_0=1$]{\includegraphics{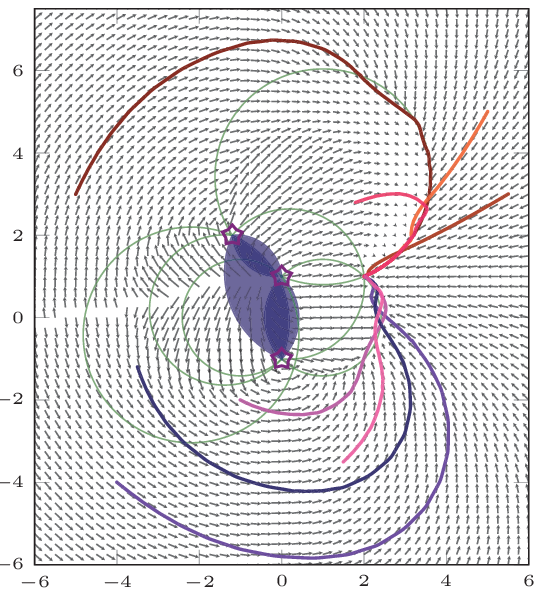}\label{fig:doubleintegrator}} \hfill
\subfloat[Unicycle, $k_\upsilon=1, k_\omega=2$]{\includegraphics{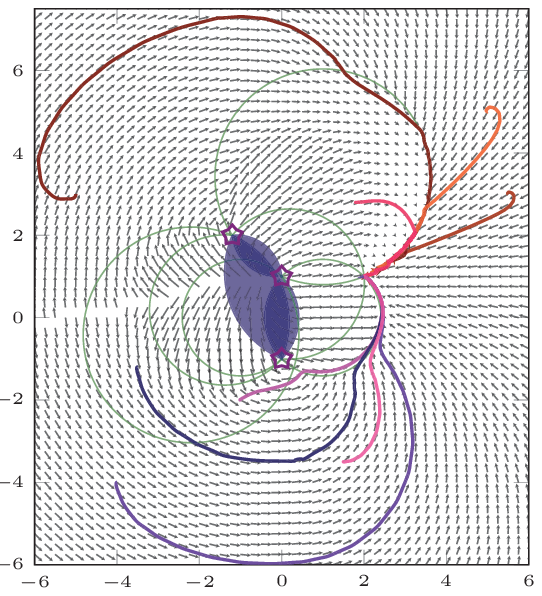}\label{fig:unicycle}} \hfill
\subfloat[Double integrator, $\lambda_0=1$]{\includegraphics{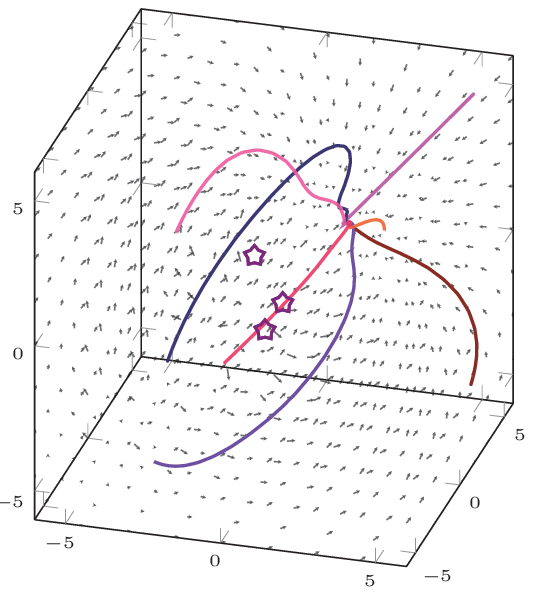}\label{fig:unicycle}} \hfill
\caption{
Plot of trajectories of a double integrator and a unicycle for the visual homing problem with 3 landmarks from different starting points. The goal positions are located at $\vx^*=[2,1]^\trp$ and $\vx^*=[2,1,2]^\trp$, alongside the vector field $\bm f(\vx)$.}
\label{fig:simulation}
\end{figure*}
\begin{figure}[b]
\subfloat[Double integrator]{\scalebox{1}{\includegraphics{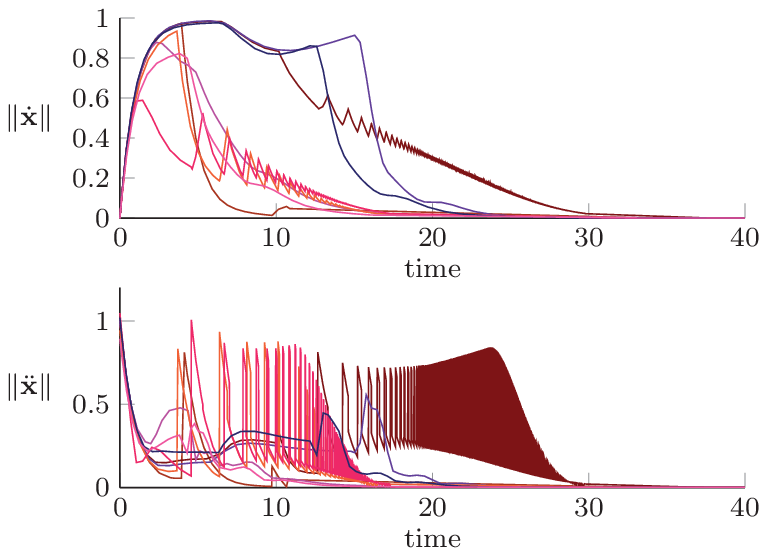}}\label{fig:f1}}\\
\subfloat[Unicycle]{\scalebox{1}{\includegraphics{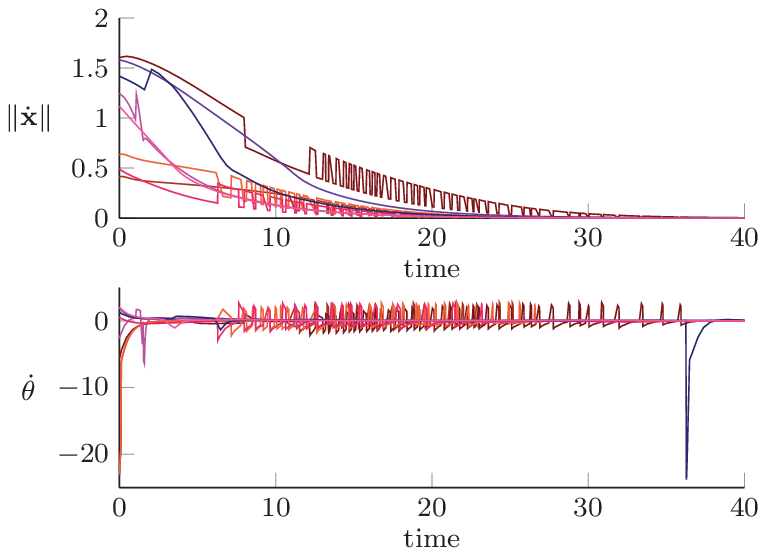}}\label{fig:f2}}
\caption{
In~\protect\subref{fig:f1}, the magnitude of velocity and acceleration of the double integrator systems from Fig~\ref{fig:doubleintegrator} are given. In~\protect\subref{fig:f2}, the linear and angular velocity of trajectories from Fig~\ref{fig:unicycle} are plotted.}
\end{figure}
The normal gain function $g_n$:
\begin{enumerate*}
\item pushes the robot away if $\delta_{\hat\imath\hat\jmath}$ is negative by $b_\epsilon(-\delta_{\hat\imath\hat\jmath})$, \item absorbs the robot towards the foci by $b_\epsilon(\delta_{\hat\imath\hat\jmath})\cos(\measuredangle(\vv,\vv^*))$ once $\measuredangle(\vv,\vv^*)$ is less than $\frac{\pi}{2}$.
\end{enumerate*}
Employing a bump function will smoothen the transition of the normal flow at $\cD^*$. We choose $\epsilon$ such that $\epsilon<\min_{p,q}\vu_p^{*\trp}\vu_q^* -\cos(\phifov)$, to ensures that the smooth transition phase does not fall in any obstacle set $\cO_{ij}$. For the such pair $p,q$ with $\alpha=\vu_p^{*\trp}\vu_q^* -\cos(\phifov)$, we have $b_\epsilon(-\delta_{pq})=1$ if $-\alpha<\delta_{pq}\leq-\epsilon$, meaning that the normal function is pushing away with maximum strength after the smooth transition phase (when $-\epsilon<\delta_{pq}<0$) ends.
\begin{remark}
The visual homing assumes that the desired bearings $\{\vu_i^*\}_{i=1}^k$ are given, however, $\bm f$ and $\bm f^\circ$ only requires the unit vector $\vv^*$ and the angle between the bearings of two of the desired landmarks. In fact, the pair $\vv^*$ and $\measuredangle(\vu_i^*,\vu_j^*)$ for a single pair $i,j$) produce a minimal representation of the home location $\vx^*$ (See Fig.~\ref{fig:elliptical_flow3}).
\end{remark}

\subsection{Double Integrator Control Synthesis}
With $\bm f(\vx)$ at hand, we now proceed to control design for a moving robot with a second-order integrator model. We assume the following linear system dynamics:
\begin{equation}
\ddot{\vx}=-\lambda_0\dot{\vx}+\bm\mu
\end{equation}
where $\vx\in\real{d}$ is the position of the robot, and $\lambda_0> 0$ is a small damping coefficient (either synthetically enforced or by natural viscous drag). The control input is given by:
\begin{equation}\label{eq:lawdi}
\bm\mu=\bm f(\vx).
\end{equation}
The control law in $\eqref{eq:lawdi}$ forces the double integrator system to follow $\bm f(\vx)$. While momentum and harsh initial conditions could result in violations of the FOV obstacles sets, we observe convergence under mild initial conditions.  

\subsection{Unicycle Control Synthesis}
For unicycles, instead of using $\bm f$, we use its normalized version $\bm f^\circ$ as navigation function and manually set the velocities.
Consider the unicycle dynamics with state variables $\vq=[\,\vx^\trp, \theta\,]^\trp\in\real{3}$ consisting of the position $\vx\in\real{2}$ and the orientation $\theta$ of the robot, with equations of motion:
\begin{equation}
\dot{\vq}=[\,\cos(\theta)\quad \sin(\theta)\quad 0\,]^\trp \upsilon + [\,0\quad0\quad 1\,]^\trp\omega\, ,
\end{equation}
where $\upsilon,\omega\in\real{}$ are linear and angular velocities of the robot with respect to its body-fixed frame. We use the following control law:
\begin{subequations}
\label{eq:unicycle}
\begin{align}
\upsilon &= k_\upsilon\big(\sqrt{1-\vv^\trp\vv^*}+|\delta_{ij} |\big)  \label{eq:unispeed}\\
\omega &= -k_\omega(\theta-\psi)+\dot\psi \label{eq:uniomega}
\end{align}
\end{subequations}
where $\psi\defeq\arctan(f^\circ_2, f^\circ_1)$ is the orientation of the vector field $\bm f^\circ(\vx)$ at robot's current location and $k_\upsilon,k_\omega$ are positive gains. Equation \eqref{eq:uniomega} yields exponential convergence of $\theta$ to $\psi$ and enables the robot to follow the integral curves of $\bm f^\circ$ to $\vx^*$.
Similar to the double integrator system, starting from harsh initial conditions, which in this case is starting very close to the FOV obstacle sets while facing them could yield violations of FOV constraints. However, under mild assumptions, convergence is achieved. 




\section{SIMULATION RESULTS}
We present simulation results for the visual homing problem for unicycles in 2D and double integrators in 2D and 3D. In Fig.~\ref{fig:simulation}, the trajectories from control laws in \eqref{eq:lawdi} and \eqref{eq:unicycle} are plotted from different starting states. Given that the initial conditions are mild, i.e. not too close to the FOV obstacle sets or at a high initial velocity towards them, convergence to the home location is achieved.

%

\section{CONCLUSIONS}
We presented a novel navigational vector field suitable for steering double integrators and unicycles for the visual homing problem. Our vector field works with a minimal representation of the home location, and is almost globally stable while avoiding the violation of field of view constraints. An interesting future direction is to use our vector field in the formation control problem in conjunction with control barrier functions.

\bibliographystyle{ieee}
\bibliography{bibroot}

\begin{thebibliography}{10}\itemsep=-1pt

\bibitem{argyros2001robot}
A.~A. Argyros, K.~E. Bekris, and S.~C. Orphanoudakis.
\newblock Robot homing based on corner tracking in a sequence of panoramic
  images.
\newblock In {\em Proceedings of the 2001 IEEE Computer Society Conference on
  Computer Vision and Pattern Recognition. CVPR 2001}, volume~2, pages II--II.
  IEEE, 2001.

\bibitem{connolly1990path}
C.~I. Connolly, J.~B. Burns, and R.~Weiss.
\newblock Path planning using laplace's equation.
\newblock In {\em Proceedings., IEEE International Conference on Robotics and
  Automation}, pages 2102--2106. IEEE, 1990.

\bibitem{corke2003mobile}
P.~Corke.
\newblock Mobile robot navigation as a planar visual servoing problem.
\newblock In {\em Robotics Research}, pages 361--372. Springer, 2003.

\bibitem{cowan2002visual}
N.~J. Cowan, J.~D. Weingarten, and D.~E. Koditschek.
\newblock Visual servoing via navigation functions.
\newblock {\em IEEE Transactions on Robotics and Automation}, 18(4):521--533,
  2002.

\bibitem{feder1997real}
H.~J.~S. Feder and J.-J. Slotine.
\newblock Real-time path planning using harmonic potentials in dynamic
  environments.
\newblock In {\em Proceedings of International Conference on Robotics and
  Automation}, volume~1, pages 874--881. IEEE, 1997.

\bibitem{gauss2005general}
K.~F. Gauss and P.~Pesic.
\newblock {\em General investigations of curved surfaces}.
\newblock Courier Corporation, 2005.

\bibitem{hernandez2011convergence}
E.~G. Hern{\'a}ndez-Mart{\'\i}nez, E.~Aranda-Bricaire, F.~Alkhateeb,
  E.~Maghayreh, and I.~Doush.
\newblock {\em Convergence and collision avoidance in formation control: A
  survey of the artificial potential functions approach}.
\newblock INTECH Open Access Publisher Rijeka, Croatia, 2011.

\bibitem{hong1992image}
J.~Hong, X.~Tan, B.~Pinette, R.~Weiss, and E.~M. Riseman.
\newblock Image-based homing.
\newblock {\em IEEE Control Systems Magazine}, 12(1):38--45, 1992.

\bibitem{karimian2020bearing}
A.~Karimian and R.~Tron.
\newblock Bearing-only consensus and formation control under directed
  topologies.
\newblock In {\em 2020 American Control Conference (ACC)}, pages 3503--3510.
  IEEE, 2020.

\bibitem{kemperman1987median}
J.~Kemperman.
\newblock The median of a finite measure on a banach space.
\newblock {\em Statistical data analysis based on the L1-norm and related
  methods (Neuch{\^a}tel, 1987)}, pages 217--230, 1987.

\bibitem{khatib1986real}
O.~Khatib.
\newblock Real-time obstacle avoidance for manipulators and mobile robots.
\newblock In {\em Autonomous robot vehicles}, pages 396--404. Springer, 1986.

\bibitem{kim1992real}
J.-O. Kim and P.~Khosla.
\newblock Real-time obstacle avoidance using harmonic potential functions.
\newblock 1992.

\bibitem{lambrinos1998landmark}
D.~Lambrinos, R.~M{\"o}ller, R.~Pfeifer, and R.~Wehner.
\newblock Landmark navigation without snapshots: the average landmark vector
  model.
\newblock In {\em Proc. Neurobiol. Conf. G{\"o}ttingen}, 1998.

\bibitem{lim2009robust}
J.~Lim, N.~Barnes, et~al.
\newblock Robust visual homing with landmark angles.
\newblock In {\em Robotics: science and systems}, 2009.

\bibitem{liu2010bearing}
M.~Liu, C.~Pradalier, Q.~Chen, and R.~Siegwart.
\newblock A bearing-only 2d/3d-homing method under a visual servoing framework.
\newblock In {\em 2010 IEEE International Conference on Robotics and
  Automation}, pages 4062--4067. IEEE, 2010.

\bibitem{liu2013visual}
M.~Liu, C.~Pradalier, and R.~Siegwart.
\newblock Visual homing from scale with an uncalibrated omnidirectional camera.
\newblock {\em IEEE Transactions on Robotics}, 29(6):1353--1365, 2013.

\bibitem{loizou2007biologically}
S.~G. Loizou and V.~Kumar.
\newblock Biologically inspired bearing-only navigation and tracking.
\newblock In {\em 2007 46th IEEE Conference on Decision and Control}, pages
  1386--1391. IEEE, 2007.

\bibitem{lopez2009homography}
G.~L{\'o}pez-Nicol{\'a}s, N.~R. Gans, S.~Bhattacharya, C.~Sagues, J.~J.
  Guerrero, and S.~Hutchinson.
\newblock Homography-based control scheme for mobile robots with nonholonomic
  and field-of-view constraints.
\newblock {\em IEEE Transactions on Systems, Man, and Cybernetics, Part B
  (Cybernetics)}, 40(4):1115--1127, 2009.

\bibitem{panagou2014motion}
D.~Panagou.
\newblock Motion planning and collision avoidance using navigation vector
  fields.
\newblock In {\em 2014 IEEE International Conference on Robotics and Automation
  (ICRA)}, pages 2513--2518. IEEE, 2014.

\bibitem{panagou2016distributed}
D.~Panagou.
\newblock A distributed feedback motion planning protocol for multiple unicycle
  agents of different classes.
\newblock {\em IEEE Transactions on Automatic Control}, 62(3):1178--1193, 2016.

\bibitem{papanikolopoulos1993adaptive}
N.~P. Papanikolopoulos and P.~K. Khosla.
\newblock Adaptive robotic visual tracking: Theory and experiments.
\newblock {\em IEEE Transactions on Automatic Control}, 38(3):429--445, 1993.

\bibitem{rimon1992exact}
E.~Rimon and D.~E. Koditschek.
\newblock Exact robot navigation using artificial potential functions.
\newblock {\em Departmental Papers (ESE)}, page 323, 1992.

\bibitem{tron2014optimization}
R.~Tron and K.~Daniilidis.
\newblock An optimization approach to bearing-only visual homing with
  applications to a 2-d unicycle model.
\newblock In {\em 2014 IEEE International Conference on Robotics and Automation
  (ICRA)}, pages 4235--4242. IEEE, 2014.

\bibitem{vardi2000multivariate}
Y.~Vardi and C.-H. Zhang.
\newblock The multivariate l1-median and associated data depth.
\newblock {\em Proceedings of the National Academy of Sciences},
  97(4):1423--1426, 2000.

\end{thebibliography}

\end{document}